\documentclass[12pt,oneside,english]{amsart}
\usepackage[T1]{fontenc}
\usepackage[latin9]{inputenc}
\usepackage{geometry}
\usepackage{color}
\usepackage{amsthm}
\usepackage{graphicx}
\usepackage{setspace}
\usepackage{esint}
\PassOptionsToPackage{normalem}{ulem}
\usepackage{ulem}
\makeatletter
\numberwithin{equation}{section}
\numberwithin{figure}{section}
  \theoremstyle{plain}
  \ifx\thechapter\undefined
    \newtheorem{thm}{\protect\theoremname}
  \else
    \newtheorem{thm}{\protect\theoremname}[chapter]
  \fi
  \theoremstyle{remark}
  \ifx\thechapter\undefined
    \newtheorem{rem}{\protect\remarkname}
  \else
    \newtheorem{rem}{\protect\remarkname}[chapter]
  \fi

\usepackage{lscape}
\usepackage{lineno}

\makeatother

\usepackage{babel}
  \providecommand{\remarkname}{Remark}
\providecommand{\theoremname}{Theorem}

\begin{document}

\title{Generalization of Carey's Equality and A Theorem on Stationary Population}

\maketitle
(Appeared in \emph{Journal of Mathematical Biology) }(\textbf{Springer})

\begin{center}

\textbf{Arni S.R. Srinivasa Rao (Corresponding Author)}

Georgia Regents University

1120 15th Street, Augusta, GA 30912, USA

Email: arrao@gru.edu

\vspace{0.6cm}

\textbf{James R. Carey}

Department of Entomology

University of California, Davis, CA 95616 USA 

and 

Center for the Economics and Demography of Aging

University of California, Berkeley, CA 94720

Email: jrcarey@ucdavis.edu

\end{center}

\tableofcontents{}

\pagebreak{}
\begin{abstract}
Carey's Equality pertaining to stationary models is well known. In
this paper, we have stated and proved a fundamental theorem related
to the formation of this Equality. This theorem will provide an in-depth
understanding of the role of each captive subject, and their corresponding
follow-up duration in a stationary population.\textbf{ }We have demonstrated
a numerical example of a captive cohort and the survival pattern of
medfly populations. These results can be adopted to understand age-structure
and aging process in stationary and non-stationary population population
models. \textbf{Key words:} Captive cohort, life expectancy, symmetric
pattern\textbf{s. }
\end{abstract}

\section{Introduction}

The motivation to explore the question from which Carey's Equality
ultimately emerged stemmed from the lack of methods for estimating
age structure in insect populations in particular, and in a wide range
of animal populations generally. Although various methods exist for
estimating the age of individual insects, including the use of mechanical
damage, chemical analysis and gene expression, many are costly, most
require major training and calibration efforts and none are accurate
at older ages. Because of these technical constraints, our concept
was to explore the possibility of using the distribution of remaining
lifespans of captured insects (e.g. fruit flies) to estimate age structure
statics and dynamics. The underlying idea was that the death distribution
of all standing populations must necessarily be related to the population
from which it is derived, \emph{ceteris perabus}. 

James Carey first observed the symmetric survival patterns of a captive
cohort and follow-up cohort of medflies (see for example in \cite{Muller2004,Carey2008}).
Later it was referred to as Carey's Equality \cite{Vaupel2009} although
a formal mathematical statement on a generalized Equality of Carey's
type was missing. The phenomena of Equality of mean length of life
lived by a cohort of individuals up to a certain time, and the mean
of their remaining length of life in a stationary population, was
well established \cite{Kim (SIAM),Goldstein}. Proof of Carey's Equality
in \cite{Vaupel2009} is based on the stationary population property
explained in \cite{Goldstein} although similar phenomena were also
found to be useful in understanding the renewal process,\textbf{ }see
Chapter 5 in \cite{Cox}. We have formally stated a fundamental theorem
to justify Carey's Equality under a stationary population model, which
would allow us to obtain expected lifetimes of a captive cohort of
subjects in which each subject is captured at a different age, starting
at birth. Further, we have proved our statement under a certain combinatorial
and stationary population modeling set-up. These results are extended
to two dimensions for estimating age-structure of wild populations
and understanding internal structure of aging process, age-structure
of wild populations \cite{Rao&Carey2014}. One also needs to understand
the implications of our results for estimating the age structure in
stable populations \textbf{\cite{Rao-2014-Notices}.}

\section{Main Theorem}
\begin{thm}
\label{Main Theorem}Suppose $\left(X,Y,Z\right)$ is a triplet of
column vectors, where $X=\left[x_{1},x_{2},\cdots,x_{k}\right]^{T},$
$Y=\left[y_{1},y_{2},\cdots,y_{k}\right]^{T}$, $Z=\left[z_{1},z_{2},\cdots,z_{k}\right]^{T}$
representing capture ages, follow-up durations, and lengths of lives
for $k-$subjects, respectively. Suppose, $F\left(Z\right)$, the
distribution function of $Z$ is known and follows a stationary population.
Let $G_{1}$ be the graph connecting the co-ordinates of $S_{Y}$,
the survival function whose domain is $N(k)=\{1,2,3,...,k\}$ i.e.
the set of first $k$ positive integers and $S_{Y}\left(j\right)=y_{j}$
for $j=1,2,...,k.$ Let $G_{2}$ be the graph connecting the co-ordinates
of $C_{X},$ the function of capture ages whose domain is $N(k)$
and $C_{X}\left(j\right)=x_{j}$ for all $j=1,2,...,k.$ Suppose $C_{X}^{*}\left(-j\right)=x_{j}$
for all $j=1,2,...,k.$ Let $\mathcal{H}$ be the family of graphs
constructed using the co-ordinates of $C_{X}^{*}$ consisting of each
of the $k!$ permutations of graphs. Then one of the members of $\mathcal{H}$
$(say,$$H_{g}$) is a vertical mirror image of $G_{1}.$ \end{thm}
\begin{proof}
\textcolor{black}{In the hypotheses, the information on capture ages,
follow-up durations, and lengths of lives are given as column vectors.}\textbf{\textcolor{black}{{}
}}\textcolor{black}{For analyzing data collected from life table studies
on different organisms (including life expectancy), column vectors
are used for representing a convenient data structure, where the matrix
of interest has only one column}\textbf{\textcolor{black}{.}}\textcolor{black}{{}
Capture ages of all subjects in our paper are arranged in a column
such that it will be suitable to do addition operations on these ages
with a column vector consisting of lives left for corresponding subjects.
One also can compute inner product of these two column vectors, but
the resultant scalar obtained by such operation is not of interest
in the present context. Life table analysis in demography considers
column-wise information and each column has a special purpose even
though the information obtained from the last column has the key data
needed for actuarial modeling, mortality analysis, etc, Using these
column vectors of capture ages and follow-up durations from the hypothesis,
we have constructed two sets, $I_{1}$ and $I_{2}$.}\textcolor{magenta}{{} }

Let $I_{1}=\left\{ \left(i,y_{i}\right):1\leq i\leq k\right\} $ and
$I_{2}=\left\{ \left(i,x_{i}\right):1\leq i\leq k\right\} .$ $G_{1}$
is constructed using specific ordered pairs of $I_{1}$, explained
later in the proof, and $G_{2}$ is constructed by corresponding ordered
pairs of $I_{2}.$ There are two criteria, $U_{1}$ and $U_{2},$
that govern the one-to-one correspondence properties of the two functions,
$S$ and $C_{X}$. 

\begin{eqnarray*}
U_{1} & = & \mbox{\ensuremath{\left\{  \mbox{No two subjects are captured who are of the same age and}\right.} }\\
 &  & \mbox{\ensuremath{\left.\mbox{there are no two subjects whose follow-up times until death are identical}\right\} } . }
\end{eqnarray*}

\begin{eqnarray*}
U_{2} & = & \mbox{\ensuremath{\left\{  \mbox{There is more than one subject which has the same capture ages, and}\right.} }\\
 &  & \mbox{\ensuremath{\left.\mbox{there is more than one subject which has the same follow-up durations until death}\right\} } . }
\end{eqnarray*}

Suppose $f_{1}$ and $f_{2}$ are probability density functions of
capture times and follow-up times to death, then by \cite{Vaupel2009},
we will have $f_{1}(a)=f_{2}(a)$, which is the probability of an
individual who lived $'a'$ years is the same as probability that
this individual lives $'a'$ years during the follow-up \cite{Vaupel2009}.
For an infinite population in a continuous time set-up for each $x\in C_{X}$
there exists a $y\in S_{Y}$ such that 

\begin{eqnarray}
\int_{0}^{\infty}x(s)ds+\int_{0}^{\infty}y(s)ds & = & \int_{0}^{\infty}z(s)ds\label{eq:x+y=00003Dz}\\
\nonumber 
\end{eqnarray}

and 

\begin{eqnarray}
x_{A} & =y_{A} & =\frac{1}{2}\int_{0}^{\infty}z(s)ds,\label{eq:xA}\\
\nonumber 
\end{eqnarray}

where $x_{A}$ and $y_{A}$ are the average capture ages and average
follow-up durations until death. Sum of the age at capture of $i^{th}-$
subject (i.e. $x_{i}$) and the follow-up duration (or time left in
the current context) of $i^{th}-$subject (i.e. $y_{i}$) is equal
to the total length of the life (i.e. $z_{i}$) for the $i^{th}-$subject.
$ $\textbf{ }See\textbf{ \cite{Vaupel2009,Ryder} }and\textbf{ }Chapter
3 in \cite{Preston et al} for a description of stationary population
models and\textbf{ }see page 49 in \cite{Goswami&Rao-book} and see
page 52 in \cite{Lawler}\textbf{ }for the stationary distributions.\textbf{ }

Let's define $\max\left\{ y_{i}\right\} _{i=1}^{k}=y_{t_{1}}$ for
some $t_{1}^{th}-$ subject out of $k-$subjects,

$\max\left\{ y_{i}:\mbox{for all }i\in N(k)\mbox{ and }i\neq t_{1}\right\} =y_{t_{2}}$
for some $t_{2}^{th}-$subject out of $k-1$ subjects and so on until
we arrive at, 

$\max\left\{ y_{i}:\mbox{for all }i\in N(k)\mbox{ and }i\neq t_{1},t_{2},...,t_{(k-1)}\right\} =y_{t_{k}}$.

$G_{1}$ is constructed by joining the co-ordinates of the set $S$
on the first quadrant, where

\begin{eqnarray}
S & = & \left\{ \left(t_{1},y_{t_{1}}\right),\left(t_{2},y_{t_{2}}\right),\cdots,\left(t_{k},y_{t_{k}}\right)\right\} .\label{eq:S}\\
\nonumber 
\end{eqnarray}

We call each co-ordinate of $S$ as a cell of $S.$ \textcolor{black}{One
can visualize, cells in S are made up of ordered pair of co-ordinates,
where abscissa is the subject captured and ordinate is the life left
for this subject after capture. Here, the graph, we mean by a curve
obtained by joining the cells in $S$. Each cell, except the first
and the last cell of S are joined to both sides of its neighboring
cells.} In case there is more than one subject with a maximum value
at one or more stage above, it will lead to two or more identical
co-ordinates that are used in $G_{1}.$ When $S$ satisfies $U_{1},$
$G_{1}$ is a graph of decreasing\textbf{ }function\textbf{; }when
$S$ is satisfies criterion $U_{2},$ $G_{1}$ is a\textbf{ }graph
of combination of decreasing and non-decreasing functions\textbf{. }

We can construct a sequence of quantities $x_{u_{1}},$$x_{u_{2}},\cdots,$$x_{u_{k}}$
similarly to $S$ to form the set $T$, given below:

\begin{eqnarray}
T & = & \left\{ \left(u_{1},x_{u_{1}}\right),\left(u_{2},x_{u_{2}}\right),\cdots,\left(u_{k},x_{u_{k}}\right)\right\} \label{eq:T}
\end{eqnarray}
Corresponding to $\left(t_{1},y_{t_{1}}\right)\in S$ there exists
a $\left(j,x_{j}\right)\in I_{2}$ which could be $\left(u_{1},x_{u_{1}}\right)$
or some other cell in $T.$ Note that, $C_{X}^{*}=\left\{ \left(-j,x_{j}\right):1\leq j\leq k\right\} $.
Corresponding to $\left(t_{1},y_{t_{1}}\right)\in S$ if there exists
a cell $\left(\left|-i\right|,x_{i}\right)\in C_{X}^{*}$ such that
$\left|t_{1}-\left|-i\right|\right|=0$ and $\left|y_{t_{1}}-x_{i}\right|=0$
then we call $\left(t_{1},y_{t_{1}}\right)$ and $\left(\left|-i\right|,x_{i}\right)$
are a pair of equidistant cells from a vertical axis. We can construct
a graph $H_{1}$ using cells of $C_{X}^{*}$ in the natural order
of integers. In total, we can have $k!$ permutations of orders to
construct $H_{1,\cdots,}H_{k!}$. Suppose we denote the first combination
of cells as $\left\{ \left(-j^{(1)},x_{j}^{(1)}\right):1\leq j\leq k\right\} $
and the second combination of cells as $\left\{ \left(-j^{(2)},x_{j}^{(2)}\right):1\leq j\leq k\right\} $
and so on. \textcolor{black}{Negative index here indicates that the
subject captured is considered on the negative $x-$axis, which is
similar to the left part of the Figure}\textbf{\textcolor{magenta}{{}
\ref{main figure},}}\textcolor{black}{{} if we consider the values
25, 50 75 and 100 for pre-capture segment as -25, -50, -75 and -100
as we visualize them on the negative $x-$axis.} Thus by previous
arguments, a family of graphs $\mathcal{H}$ is constructed. One of
these combinations, for example, the $g^{th}$ combination, is used
to construct a graph which we denote with $H_{g},$ which satisfies
following Equality:

\begin{eqnarray*}
\left|t_{1}-\left|-1\right|^{(g)}\right|=0 & \mbox{and} & \left|y_{t_{1}}-x_{1}^{(g)}\right|=0\\
\left|t_{2}-\left|-2\right|^{(g)}\right|=0 & \mbox{and} & \left|y_{t_{2}}-x_{2}^{(g)}\right|=0\\
\vdots &  & \vdots\\
\left|t_{k}-\left|-1\right|^{(g)}\right|=0 & \mbox{and} & \left|y_{t_{k}}-x_{k}^{(g)}\right|=0
\end{eqnarray*}

here $\left\{ \left(\left|-1\right|^{(g)},x_{1}^{(g)}\right),\left(\left|-2\right|^{(g)},x_{2}^{(g)}\right),\cdots,\left(\left|-k\right|^{(g)},x_{k}^{(g)}\right)\right\} $
are $g^{th}$ combination of cells such that $H_{g}$ is a vertical
mirror image of $G_{1}.$ \textcolor{black}{Image of $G_{1}$ is visualized
as if it is seen from the mirror kept on $y$-axis. Note that, we
have generated $k!$ graphs by our construction, and one of such graph,
which we called as $H_{g}$ is shown to have vertical mirror image
of $G_{1}$.}
\end{proof}
\textcolor{black}{Our theorem establishes existence of graph depicting
mirror images of the pre- and post-capture longevity distributions
in Figure}\textbf{\textcolor{black}{{} \ref{main figure}}}\textcolor{black}{.
In general, in the area of population biology where survivorship curves
of captive cohort obey Carey's Equality, there is a need for understanding
the pattern within the data structure. The results by \cite{Vaupel2009,Goldstein}
explain functional symmetries as an application of renewal theory.
We do not depend on any of the classical works on renewal equation
and renewal theory proposed by Lotka \cite{Lotka-1909,Lotka-1939a,Lotka-1939b,Lotka-book}
and by Feller \cite{Feller-1941} (also see Chapter XI in \cite{Feller-book}).
Our inspiration is purely from experimental observations demonstrated
by James Carey and a statement on stationary populations in the equation
(1) in the paper by Vaupel \cite{Vaupel2009}. However, our theory
and method of proof uses sequentially arranged data of captive individuals,
which was also usually done in renewal theory analysis or proving
renewal-type of equations. }

\textcolor{black}{Hence our method provide an alternative and independent
approach for such kind of sequentially arranged captive subjects.}\textbf{\textcolor{black}{{}
}}\textcolor{black}{Besides relating the captive ages and corresponding
follow-up durations of subjects in a stationary population, the principles
of the theorem helped to visualize person-years in a follow-up starting
at birth, in a stationary population model in which subjects of each
age are captured. }

\subsection{Life expectancy }

First we construct the age structured survival function using $l(s)$,
the number of captured subjects at age $s$ and $l(s_{0})$, the number
of subjects at the beginning of the follow-up. Suppose $s_{0}$ is
the time at the beginning of the follow-up, $s_{i}$ is the time at
the $i^{th}$ time point of observation for $i=1,2,\cdots,k$. Suppose
each of the $y_{t_{i}}$ for $i=1,2,\cdots,k$ fall exactly in one
of the time intervals $\left(s_{i-1},s_{i}\right)$, then the expectancy
of life is $\frac{k+1}{2}$. The probability of death over the time
period are $q(s_{i})=\frac{1}{k-i}$ with the survival pattern, $l(s_{i})=k-i$
for $i=1,2,\cdots,k-1$ and $l(s_{k})=0$, $q(s_{k})=1.$ Suppose
there are $n_{i}$ number of $y_{t_{i}}$ falling within $\left(s_{i-1},s_{i}\right)$
such that $\sum_{i=1}^{k}n_{i}=k$ and $l(s_{i})$ follows the previous
construction. If $n_{i}>1$ in one or more of the $\left(s_{i-1},s_{i}\right)$
then there must exist empty cells where the event of death is avoided.
Let us define a number $c_{1}$ as follows:

\begin{eqnarray*}
c_{1} & = & \mbox{\ensuremath{\left\{  \mbox{Number of cells }\ensuremath{(s_{i-1},s_{i})}\mbox{ for }\mbox{ }\ensuremath{i=1,2,\cdots,k}\mbox{ where exactly one of the }\mbox{ }y_{t_{i}}\mbox{ falls }\right\} } }\\
\end{eqnarray*}

If $c_{1}=k$, then the life expectancy is $\frac{k+2}{2}$. If $c_{1}\neq k$,
then at least one of the $(s_{i-1},s_{i})$ is empty. There could
be several combinations of distribution of $y_{t_{1}}$ in $(s_{i-1},s_{i})$
when the event $c_{1}\neq k$ occurs, and we explain in the following
remark one possible situation in which deaths are concentrated in
the early ages and at late ages. Other combinations can be evaluated
using similar constructions.
\begin{rem}
Suppose $n_{1}=n_{2}=\cdots=n_{r_{1}}=1$; $n_{r_{1}+1}=n_{r_{1}+2}=\cdots=n_{r_{2}}=0$;
and $n_{r_{2}+1}=n_{r_{2}+2}=\cdots=n_{k}=1$ such that $\sum_{i=1}^{k}n_{i}=k$
for some $1<r_{1}<r_{2}<k.$ The probability of death at various ages,
$q(s_{i})$, is,

\begin{eqnarray}
q(s_{i}) & = & \left\{ \begin{array}{cc}
\frac{1}{k-i} & \mbox{ for }i=0,1,\cdots,(r_{1}-1)\\
\\
0 & \mbox{ for}\mbox{\ensuremath{}}i=r_{1},(r_{1}+1),\cdots,(r_{2}-1)\\
\\
\frac{1}{k-(r_{1}+i)} & \mbox{ for }i=0,1,\cdots(k-2)\\
\\
1 & \mbox{ for }i=(k-1)
\end{array}\right.\label{eq:q(si)}\\
\nonumber 
\end{eqnarray}

The life expectancy, $e(s_{0})$, is obtained by the formula below:

\begin{eqnarray}
e(s_{0}) & = & \frac{1}{2k}\left[\sum_{i=1}^{2r_{1}-1}\left\{ 2k-i\right\} +\sum_{j=2r_{1}+1}^{2k-1}\left\{ 2k-j\right\} \right]+\frac{\left(k-r_{1}\right)\left(r_{2}-r_{1}-1\right)}{k}\nonumber \\
\label{eq:e(c0)}
\end{eqnarray}

\end{rem}

\subsection{Person-years and means}

Under the above set-up, in this section we derive the mean age of
the captive cohort in terms of the mean of the person-units followed.
Suppose $a(k,x)$ denotes $k^{th}-$subject ($k>0)$ is captured at
age $x$, then $\int_{1}^{\infty}a(k,x)dk$ is number of subjects
captured at age $x.$ The mean age at capture for all subjects of
the cohort formed of subjects of all ages, $c(0)$, is

\begin{eqnarray}
c(0) & = & \frac{\int_{0}^{\infty}\left[x\int_{1}^{\infty}a(k,x)dk\right]dx}{\int_{0}^{\infty}\int_{1}^{\infty}a(k,x)dkdx}.\label{eq:capture mean age}\\
\nonumber 
\end{eqnarray}

Let $\int_{0}^{s}b(y)dy$ be the number of deaths out of $\int_{0}^{\infty}\int_{1}^{\infty}a(k,x)dkdx$
during the age $0$ to $s$. The number of subjects surviving at age
$s$ is 

\begin{eqnarray*}
 &  & \int_{0}^{\infty}\int_{1}^{\infty}a(k,x)dkdx-\int_{0}^{s}b(y)dy\\
\end{eqnarray*}

and the number of subjects surviving at age $(n+1)s$ for some positive
integer $n$ is 

\begin{eqnarray*}
 &  & \int_{0}^{\infty}\int_{1}^{\infty}a(k,x)dkdx-\sum_{n=0}^{\infty}\int_{ns}^{(n+1)s}b(y)dy.\\
\end{eqnarray*}

Thus, life expectancy at the formation stage of a captive cohort\textbf{,}
say, $E[c(0)]$ can be computed by the formula

\begin{eqnarray}
E[c(0)] & = & \frac{\int_{0}^{\infty}\int_{0}^{\infty}\int_{1}^{\infty}a(k,x)dkdxds-\int_{0}^{\infty}\left(\sum_{n=0}^{\infty}\int_{ns}^{(n+1)s}b(y)dy\right)ds}{\int_{0}^{\infty}\int_{1}^{\infty}a(k,x)dkdx}\label{eq:Life Exp (followup)}\\
\nonumber 
\end{eqnarray}
We state a theorem for the total person-years to be lived by all the
subjects in a birth cohort. 
\begin{thm}
\label{Second-Theorem}Suppose subjects of each age of life in a population
are captured. Then, using the constructions in $c(0)$ and $E[c(0)]$,
the total person-years, say, $T\left(a,c(0),E[c(0)]\right)$, that
will be lived by newly born subjects in a stationary population can
be expressed as 
\begin{eqnarray*}
T\left(a,c(0),E[c(0)]\right) & = & \left(c(0)+E[c(0)]\right)\int_{0}^{\infty}\int_{1}^{\infty}a(k,x)dkdx.\\
\end{eqnarray*}
\end{thm}
\begin{proof}
We have,

\begin{eqnarray}
\left(c(0)+E[c(0)]\right)\int_{0}^{\infty}\int_{1}^{\infty}a(k,x)dkdx & = & \int_{0}^{\infty}\left[x\int_{1}^{\infty}a(k,x)dk\right]dx+\nonumber \\
 &  & \int_{0}^{\infty}\int_{0}^{\infty}\int_{1}^{\infty}a(k,x)dkdxds-\int_{0}^{\infty}\left(\sum_{n=0}^{\infty}\int_{ns}^{(n+1)s}b(y)dy\right)ds\nonumber \\
\nonumber \\
 & = & \int_{0}^{\infty}\left[x\int_{1}^{\infty}a(k,x)dk\right]dx+\nonumber \\
\nonumber \\
 &  & \int_{0}^{\infty}\left[\int_{0}^{\infty}\int_{1}^{\infty}a(k,x)dkdx-\sum_{n=0}^{\infty}\int_{ns}^{(n+1)s}b(y)dy\right]ds\label{eq:RHS in proof}
\end{eqnarray}

The R.H.S. of (\ref{eq:RHS in proof}) is sum of ages of subjects
of all ages in a captive cohort, and total person-years to be lived
by the captive cohort, which is $T\left(a,c(0),E[c(0)]\right)$.
\end{proof}

\section{History and Related Results}

An Equality arising out of symmetries of life lived and life left
was named as Carey's Equality by James Vaupel \cite{Vaupel2009},
to highlight the discovery of certain symmetries in his biodemographic
experiments by James Carey (see, for example, \cite{Muller2004,Carey2008}).
Vaupel \cite{Vaupel2009} and Goldstein \cite{Goldstein} have proved
equality on such symmetries as a direct application of renewal theory.
Our main theorem in this article is not inspired by renewal theory,
but we conceptualized our approach directly from experimental results
shown by James Carey and then used equation (1) from Vaupel \cite{Vaupel2009}
in our hypothesis. Renewal theory has long history even before the
seminal works on population dynamics by Alfred Lotka (see, for example,
\cite{Lotka-1909,Lotka-1939a,Lotka-1939b}), who has used an integral
equation of type (\ref{eq:renewal-lotka}) to link number of births
at time $t>0$ with number of births a women has at time $t=0.$ 

\textbf{
\begin{equation}
B(t)=\int_{0}^{t}G(t,a)da+H(t)\label{eq:renewal-lotka}
\end{equation}
}

where $B(t)$ is total number of births at time $t$, $G(t,a)$ is
number of births from a women who is at age $a$ and alive at time
$t$ and $H(t)$ is number of births from a women who is alive at
$t=0.$ Usually, we compute $B(t)$ from $a=\alpha$ to $a=\beta$,
where $\alpha$ is lower reproductive age and $\beta$ is upper reproductive
age. $G(t,a)$ can be written as,

\begin{equation}
G(t,a)=B(t-a)s(a)\label{eq:b(t-a)s(a)}
\end{equation}
where $B(t-a)$ is total number of births at time $(t-a)$ and $s(a)$
is chance of surviving to exact age $a.$ \ref{eq:renewal-lotka}
is also referred as renewal equation. Combining \ref{eq:renewal-lotka}
and \ref{eq:b(t-a)s(a)}, we can write $B(t)$ within reproductive
ages as,

\begin{equation}
B(t)=\int_{\alpha}^{\beta}B(t-a)s(a)da\label{eq:renewal}
\end{equation}

William Feller \cite{Feller-book} provided foundations of renewal
theory in his book (see Chapters VI and XI) as did authors of other
books which were written exclusively on renewal theory (see for example,
\cite{Cox}) or contained chapters devoted to basic renewal theory
(for example, see Chapter 6 in \cite{Lawler} and Chapter 12 in \cite{Karlin-book}).
Several applications of convolutions of independent random variables
which we see in renewal theory can also be found in understanding
disease progression between one stage to another stage, epidemic prediction
and so forth (for example, see \cite{Brookmeyer and Gail} and \cite{Rao-RMJM}).\textbf{ }

\section{Example and visualization}

We provide here a practical application based on the medfly population
with a visual depiction of Carey\textquoteright{}s Equality (Figure
\ref{main figure}). Note the symmetry of the pre- and post-capture
segments of the lifespans of individuals in the population that underlies
the equivalency of life-days which, in turn, underlies the equivalency
noted by \cite{Vaupel2009}, \textquotedbl{}If an individual is chosen
at random from a stationary population...then the probability the
individual is one who has lived a years equals the probability the
individuals is one who has that number of years left to live.\textquotedblright{}\textbf{ }

Carey's Equality is important because it both builds on and complements
the properties of one of the most important models in demography:
the stationary population model. Stationary population model is fundamental
to formal demography because it is a special case of the stable population
model, and provides explicit expressions that connect the major demographic
parameters to one another, including life expectancy, birth rates,
death rates and age structure, see Chapter 3 in \cite{Preston et al}.
A graphical depiction of Carey's Equality, shown in Figure \ref{fig2},
shows the interconnectedness of these parameters. Note that the shape
of the stationary population \ref{fig2}(a) as well as both its age
and death distributions \ref{fig2}(b) are identical, and that the
proportion of each age class in the captive population \ref{fig2}(c)
is identical to the proportions within the whole population (i.e.
due to the symmetrical distributions of pre- and post-capture lifespans
in \ref{fig2}(b)).

\section{\textcolor{black}{Discussion}}

\textcolor{black}{The approach we used in this paper to demonstrate
the mathematical identity underlying Carey's Equality is fundamentally
different than that used by previous authors. Originally James Carey
created a simple life table model (Table 1 in \cite{Muller2004})
to demonstrate the equality of age structure and post-capture deaths
in a stationary population, the identity of which was then formalized
by statisticians Hans Müller and Jane-Ling Wang \cite{Muller2004}.
Vaupel \cite{Vaupel2009} and Goldstein \cite{Goldstein} followed
by using mathematical first principles to derive the equality. }

\textcolor{black}{Our approach differs from the one we just outlined.
Instead of first formulating and then deriving the equality, we justified
the constancy of post-capture patterns of death in stationary medfly
populations in the following two steps. The first was to conjecture
that in stationary populations the ordered pre- and post-capture life
course segments of individuals will be symmetrical (Theorem \ref{Main Theorem})
and that total person-years in a captive cohort can be formulated
(Theorem \ref{Second-Theorem}). The second was to prove these relationships
using a series of mathematical assumptions (i.e. the respective proofs). }

\textcolor{black}{Our approach contributes to the demographic literature
in general and specifically to an understanding of Carey's Equality
in several ways. First, our theorems provide an independent method
for formulating a mathematical relationship in formal demography.
Indeed we are unaware of any other models in formal demography that
have involved proofs from mathematical conjectures (theorems). Second,
our proofs allow us to state unequivocally that the Carey Equality
will be true in all stationary populations. Although this is a logical
outcome from all of the earlier approaches, our approach makes this
result both explicit and conclusive. Third, our proofs required that
we draw on set theory, an area of mathematics involving logic that
is not commonly used in demography. As a consequence of the problem
framed in fundamental mathematics, Carey's Equality is better situated
both to draw from and be extended into other areas of basic mathematical
theory. Fourth, using the ideas of the main theorem, we are positioned
to obtain further results related to Carey's Equality such as for
higher dimensions and probabilistic and deterministic results for
multiple captive cohorts. }

\begin{figure}
\includegraphics[scale=0.7]{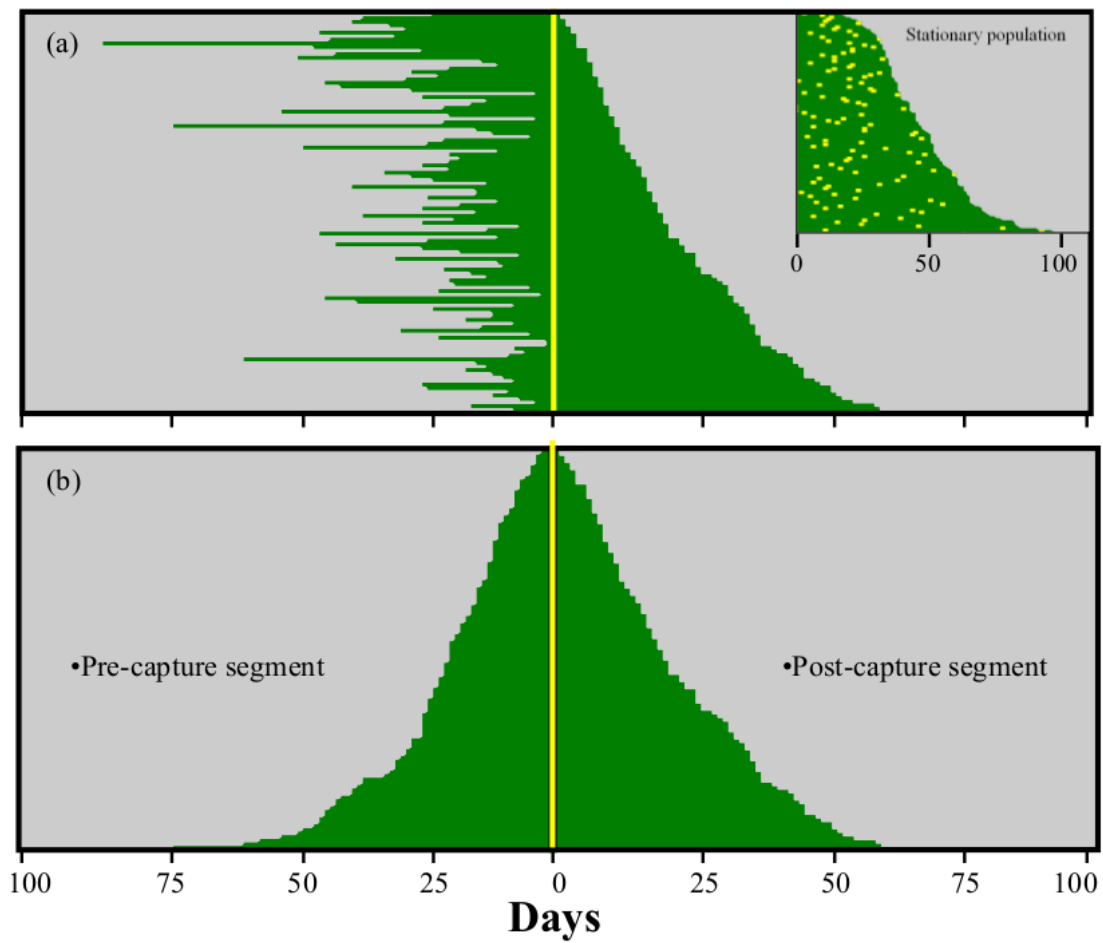}

\caption{\label{main figure}Schematics of a randomly sampled hypothetical
stationary Mediterranean fruit fly population (adapted from \cite{Carey2012}).
Horizontal lines depict the life course of individual flies from birth
(eclosion) to death. (a) Example of the life courses of captive individuals
divided into pre- and post-capture segments and ordered from shortest-to-longest
(top-to-bottom) post-capture lifespans. Inset shows the stationary
medfly population in which the yellow tick marks depict the ages at
which individual medflies are sampled. (b) Same as (a) except with
the pre- and post-capture segments of individuals decoupled and both
ordered according to length from top-to-bottom to show the symmetry
(mirror image) of the distributions. }

\end{figure}

\begin{figure}
\includegraphics[scale=0.8]{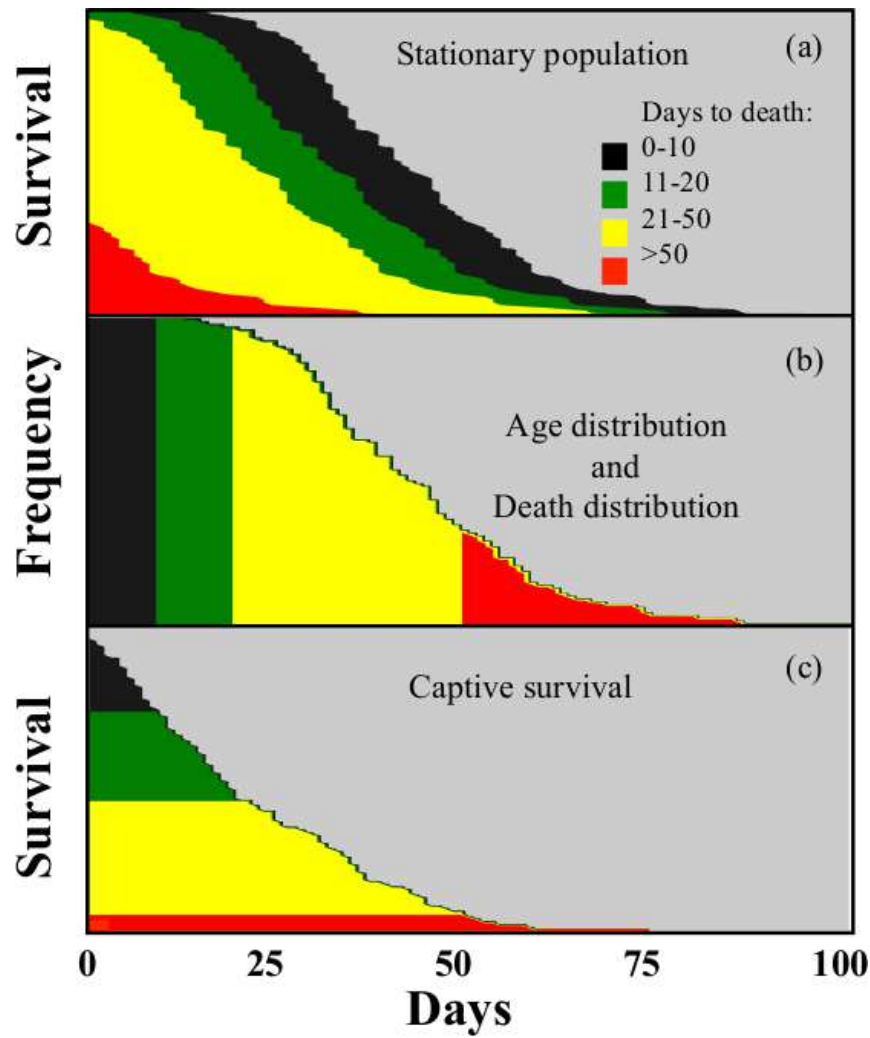}

\caption{\label{fig2}Structural interconnectedness of an hypothetical stationary
Mediterranean fruit fly population showing: (a) days to death by age
groups; (b) age and death distributions; and (c) captive survival. }
\end{figure}

\section{Conclusions}

\textcolor{black}{Our paper offers new sets of tools, techniques and
theoretical framework in terms of visualization of the demographic
data involving capture ages and development of new theoretical ideas
for analyzing data obtained by captive cohorts. Such approaches will
have applications in other demographic situations, for example, understanding
aging patterns in a captive cohort data when information on lives
left is right truncated, projecting various scenarios of demographic
transition, etc, }

The main result of our paper will be useful in understanding the relationship
between average lengths of lives of captured, follow-up, and total
lengths of\textbf{ }the lives in a stationary population\textbf{.}
Our method of re-structuring the follow-up durations of captive cohort
can be adopted also for non-stationary populations, which can be used
for understanding the internal structures of the population with respect
to the age at capture. This will enable us to look deeper into the
aging process of stationary and non-stationary populations.\textbf{
}For each captive cohort of subjects there exists an associated, exact
configuration of a combination of coordinates of the survival graphs,
and this association is dynamic. The right combination of coordinates
is dependent on the formation of the captive cohort. The idea of a
proof through formation of symmetric graphs, combined with captive
age distribution is novel. We have demonstrated the utility of such
thinking in understanding symmetric patterns formed of a captive cohort
and associated follow-up lengths. This strategy was also helpful for
us in deriving formulae for expectation of life in a stationary population,
discretely, and also using multiple integrals. The theory explained
here can be adopted to both human and non-human populations.

\section{Acknowledgements}

We thank the organizers of the Keyfitz Centennial Symposium on Mathematical
Demography sponsored by the Mathematical Biosciences Institute, Ohio
State University, June 2013. Research by JRC supported by NIA/NIH
grants P01 AG022500-01 and P01 AG08761-10.

\end{document}